\documentclass[aps,prd,twocolumn,superscriptaddress,nofootinbib,10pt]{revtex4-2}

\usepackage{amsmath,amssymb,mathrsfs,graphicx,bm,xcolor}
\usepackage{amsthm}
\usepackage[colorlinks=true,linkcolor=blue,citecolor=blue,urlcolor=blue]{hyperref}
\usepackage{orcidlink}
\usepackage{array}

\newtheorem{proposition}{Proposition}
\newtheorem{corollary}{Corollary}
\newtheorem{remark}{Remark}

\newcommand{\F}{F}
\newcommand{\Eeff}{\mathcal{E}_{K}}
\newcommand{\Q}{\mathcal{Q}}

\begin{document}

\title{Scalaron-modified null focusing and radial monotonicity in static \texorpdfstring{$f(R)$}{f(R)} gravity}

\author{Maickol Muñoz-Palma\orcidlink{0009-0008-1017-4406}}
\email{m.munoz58@ufromail.cl}
\affiliation{Departamento de Ciencias Físicas, Universidad de La Frontera, Casilla 54-D, 4811186 Temuco, Chile.}

\author{Francisco S. N. Lobo\orcidlink{0000-0002-9388-8373
}}
\email{fslobo@ciencias.ulisboa.pt}
\affiliation{Instituto de Astrof\'isica e Ci\^encias do Espa\c{c}o, Faculdade de Ci\^encias da Universidade de Lisboa, Edif\'icio C8, Campo Grande, P-1749-016 Lisbon, Portugal}
\affiliation{Departamento de F\'isica, Faculdade de Ci\^encias da Universidade de Lisboa, Edif\'icio C8, Campo Grande, P-1749-016 Lisbon, Portugal}

\author{Jean B\'aez Cuevas\orcidlink{0000-0002-3308-3362}}
\email{jean.baez@pucv.cl}
\affiliation{
Instituto de F\'isica, Pontificia Universidad Cat\'olica de Valpara\'iso, Casilla 4950, Valpara\'iso, Chile.}

\author{Francisco Tello-Ortiz \orcidlink{0000-0002-7104-5746}}
\email{francisco.tello@ufrontera.cl}
\affiliation{Departamento de Ciencias Físicas, Universidad de La Frontera, Casilla 54-D, 4811186 Temuco, Chile.}

\begin{abstract}
We derive an exact radial monotonicity law for static, spherically
symmetric spacetimes in metric \(f(R)\) gravity. For
\(ds^2=A(r)dt^2-dr^2/B(r)-r^2d\Omega^2\), the matter contribution and
the scalaron Hessian combine into an effective radial-convergence
numerator that fixes the derivative of \(\Q=B/A\). On every connected
static interval with \(A>0\), \(B>0\), and \(f_R\equiv df/dR>0\), its sign therefore
determines the monotonicity of \(\Q\). The integrated identity retains
the finite, generally nonzero value of \(B/A\) at a regular
nondegenerate Killing horizon; consequently, a fixed-sign convergence
condition orders the horizon endpoint ratios rather than excluding two
horizons. A zero-integral obstruction arises for equal endpoint values, including
boundaries where \(B\to0\) while \(A\) remains finite and nonzero. Saturation is equivalent to \(B/A=\mathrm{const}\), and in
vacuum requires a scalaron profile linear in the areal radius. We
illustrate the strict non-saturated branch, within the
General-Relativity sector, using the exact constant-density stellar
interior, and apply the equality and consistency diagnostics to the
constant-\(X\) and power-law solutions of Multam\"aki and Vilja. In
particular, in the Schwarzschild--de Sitter two-horizon parameter
range, one constant-\(X\) solution crosses \(f_R=0\) inside the
complete static patch, while a direct substitution into the original
radial equation exposes an unresolved exponent mismatch in the
displayed power-law family. The
results provide a model-independent static-sector diagnostic and a
precise starting point for a future horizon-regular treatment of
black-hole interiors.
\end{abstract}

\maketitle

\section{Introduction}
\label{sec:introduction}

The behavior of null congruences is central to the geometric analysis
of gravitational focusing, singularity formation, and horizon
structure. The Raychaudhuri equation provides a covariant evolution
law for the expansion of a congruence and relates its convergence to
the Ricci tensor projected along the corresponding null
direction~\cite{Raychaudhuri1955}. Together with suitable causal and
energy assumptions, it underlies the singularity theorems of Hawking
and Penrose~\cite{HawkingPenrose1970} and remains a fundamental tool
in the analysis of black-hole spacetimes~\cite{Wald1984}. Its role in
modified-gravity cosmology has also been explored directly, including
inhomogeneous \(f(R)\) backgrounds in which the curvature sector
modifies the standard focusing balance~\cite{Chakraborty2023}.

In General Relativity, the null convergence term is related directly
to the matter sector through the Einstein equations. This connection
has motivated horizon-counting and no-inner-horizon results based on
energy conditions and on the structure of the matter
fields~\cite{YangCaiLi2022,CaiLiYang2021,AnLiYang2021}. Related
theorems have been developed in gravitational theories containing
additional scalar or higher-curvature sectors, including
Einstein--Maxwell--Horndeski and charged Gauss--Bonnet
models~\cite{DeveciogluPark2022,DeveciogluPark2024}. These studies
illustrate a general feature of extended gravitational theories: once
additional dynamical degrees of freedom are present, the matter null
energy condition alone no longer exhausts the contributions governing
null focusing~\cite{Capozziello:2013vna,Capozziello:2014bqa}.

Metric \(f(R)\) gravity provides a particularly transparent framework
in which to investigate this effect. Its field equations contain
derivatives of \(f_R\equiv df/dR\), which represents the additional
scalar degree of freedom of the theory. These derivatives enter the
null projection additively with the matter stress tensor and may
reinforce or oppose ordinary matter focusing
\cite{SotiriouFaraoni2010,DeFeliceTsujikawa2010,
NojiriOdintsov2011,CapozzielloDeLaurentis2011}. The same theory has been studied extensively in cosmology and
inflation~\cite{SebastianiMyrzakulov2015}. Its static sector includes
constant- and nonconstant-curvature configurations, reconstructed
vacuum geometries, charged solutions, and scalar-supported black
holes~\cite{MultamakiVilja2006,SaffariRahvar2009,
SebastianiZerbini2011,Karakasis2021,Chaturvedi2023,
Hollenstein:2008hp,HurtadoArenas2020,Tang2021}. Topological black holes
and their thermodynamic properties have also been investigated in
\(R^2\) and \(F(R)\) gravity~\cite{Cognola:2015wqa,EslamPanah2023},
including electrodynamic configurations with a cosmological constant
and corrected thermodynamic quantities~\cite{Xu2024}.

The radial null projection is not itself new as an algebraic component
of the static field equations. In the vacuum construction of
Multam\"aki and Vilja, for example, the difference of two diagonal
field equations yields a first-order relation between
\(X\equiv A/B\) and the scalaron, and immediately gives
\(f_R''=0\) when \(X\) is constant~\cite{MultamakiVilja2006}.
Saffari and Rahvar subsequently emphasized that reconstructed spherical
solutions must be checked against the full set of independent field
equations and against any consistency relation introduced by extra
metric restrictions~\cite{SaffariRahvar2009}. More recently, Wang and Battista derived, in General Relativity,
radial energy-condition constraints on the product of the temporal and
radial metric coefficients~\cite{WangBattista2026}. In their
convention,
\(ds^2=-B_{\rm WB}dt^2+A_{\rm WB}dr^2+r^2d\Omega^2\), the monotonic
quantity is \(A_{\rm WB}B_{\rm WB}\); under the dictionary
\(B_{\rm WB}=A\) and \(A_{\rm WB}=1/B\), this becomes
\(A/B=1/\Q\).

The novelty pursued here is therefore not the claim that no radial
metric combination has previously appeared. It is the covariant
scalaron-modified convergence interpretation of the exact identity,
together with a systematic treatment of its boundary data. In
particular, when the redshift and radial functions are independent,
the relevant focusing information is distributed between the matter
sector, derivatives of \(f_R\), and the relative radial behavior of
the two metric functions. Recasting this information as a monotonicity
law for \(B/A\) provides a direct diagnostic for proposed static
solutions and makes clear exactly which boundary conclusions do, and
do not, follow.

A central issue is the distinction between two types of limiting
surface. If \(B\to0\) while \(A\) remains finite and nonzero, then
\(B/A\to0\). At a regular Killing horizon, however, both \(A\) and
\(B\) vanish, and their ratio generally approaches a finite, nonzero
limit determined by their leading radial derivatives. Any integrated
focusing relation must retain this contribution. Neglecting it would
incorrectly turn a monotonicity statement into an unconditional
restriction on the number of Killing horizons.

The purpose of the present work is to establish the corresponding
results rigorously on a connected static radial domain. We derive the
exact differential and integral identities, formulate the associated
monotonicity and endpoint theorems, and separate the matter and
scalaron contributions. We then provide a complete non-saturated
benchmark using the constant-density stellar interior, reinterpret the
constant-\(X\) vacuum class geometrically, and use the same radial
identity to test the exact solutions displayed by Multam\"aki and
Vilja. The analysis is deliberately distinct from a global
no-Cauchy-horizon theorem: the region between an outer event horizon
and an inner horizon is generally nonstatic and requires a
horizon-regular double-null formulation.

The discussion is restricted throughout to connected static intervals
on which the metric functions and \(f_R\) satisfy the regularity and
positive-coupling assumptions stated below. This setting permits a
controlled treatment of the radial problem without importing global
causal conclusions that are not contained in a static coordinate
construction.

Phrased as a single question, the
present work asks how much information about the global horizon
structure of a static configuration is already fixed by the
\emph{local} radial focusing properties of the geometry. The answer
developed below is precise but deliberately modest: within a
connected static interval, the sign of the matter-plus-scalaron
convergence numerator completely determines the monotonicity of
\(B/A\), and hence orders the finite endpoint data at any bounding
regular Killing horizons. It does not, by itself, determine whether a
second inner black-hole horizon exists, since that question concerns
the nonstatic interior rather than the static exterior. Distinguishing
sharply between what the local radial identity does and does not fix
is the central concern of this paper.

This work is organized as follows. In
Sec.~\ref{sec:field-equations}, we introduce the metric \(f(R)\) field
equations and the effective radial-convergence functional. In
Sec.~\ref{sec:null-geometry}, we construct the radial null frame and
derive the exact identity governing \(B/A\). The monotonicity and
boundary theorems are presented in Sec.~\ref{sec:theorems}. In
Sec.~\ref{sec:physical-content}, we separate the matter and scalaron
terms, discuss the General-Relativity limit, and exhibit a complete
non-saturated exact example. The equality classes and exact-solution
diagnostics are examined in Sec.~\ref{sec:applications}. Section
\ref{sec:scope} clarifies the relation to black-hole interiors, and
Sec.~\ref{sec:conclusions} summarizes the results. A direct check of
the Multam\"aki--Vilja power-law exponent relation is given in
Appendix~\ref{app:MV-power-law}.

\section{Field equations and effective radial convergence}
\label{sec:field-equations}

We consider metric $f(R)$ gravity, described by the action
\begin{equation}
 S=\frac{1}{16\pi}\int d^4x\,\sqrt{-g}\,f(R)
 +S_m[g_{\mu\nu},\Psi] .
 \label{eq:action}
\end{equation}
Variation with respect to the metric yields
\begin{equation}
 \F R_{\mu\nu}-\frac12 f g_{\mu\nu}
 -\nabla_\mu\nabla_\nu\F
 +g_{\mu\nu}\Box\F
 =8\pi T_{\mu\nu},
 \label{eq:field-equations}
\end{equation}
where $\F\equiv f_R=df/dR$,
$\Box\equiv g^{\alpha\beta}\nabla_\alpha\nabla_\beta$, and
$T_{\mu\nu}$ is the matter energy-momentum tensor.

We restrict attention to static, spherically symmetric geometries of
the form
\begin{equation}
 ds^2=A(r)\,dt^2-\frac{dr^2}{B(r)}-r^2d\Omega^2,
 \label{eq:metric}
\end{equation}
where $r$ is the areal radius and
$\F(r)=f_R\bigl(R(r)\bigr)$.
Throughout the analysis, $r>0$, and the functions
$A(r)$, $B(r)$, and $\F(r)$ are assumed to possess the differentiability
required by the field equations. Boundary quantities are understood as
one-sided limits from within the interval considered.

A \emph{static interval} $I$ is a connected open interval satisfying
\begin{equation}
 A(r)>0,\qquad B(r)>0,\qquad \F(r)>0.
 \label{eq:static-domain}
\end{equation}
The first two conditions ensure that $t$ is timelike, $r$ is spacelike,
and the radial null frame introduced below is real. The condition
$\F>0$ guarantees a positive effective gravitational coupling and is
a standard requirement in viable metric $f(R)$ models, although it
does not constitute a complete set of stability
conditions~\cite{SotiriouFaraoni2010,DeFeliceTsujikawa2010}.

The divided convergence relation and all sign theorems below depend
essentially on this assumption. A zero of \(\F\) may leave an undivided
component of the field equations algebraically meaningful, but it
prevents division by \(\F\) and invalidates any positive-coupling
monotonicity argument across that surface.

Let $K^\mu$ be a radial null vector. Contracting
Eq.~\eqref{eq:field-equations} with $K^\mu K^\nu$ eliminates the terms
proportional to $g_{\mu\nu}$ and gives
\begin{equation}
 \F R_{\mu\nu}K^\mu K^\nu
 =8\pi T_{\mu\nu}K^\mu K^\nu
 +K^\mu K^\nu\nabla_\mu\nabla_\nu\F .
 \label{eq:null-projection}
\end{equation}
We define the effective radial-convergence numerator as
\begin{equation}
 \Eeff\equiv
 8\pi T_{\mu\nu}K^\mu K^\nu
 +K^\mu K^\nu\nabla_\mu\nabla_\nu\F ,
 \label{eq:Eeff-definition}
\end{equation}
so that
\begin{equation}
 R_{\mu\nu}K^\mu K^\nu=\frac{\Eeff}{\F}.
 \label{eq:Ricci-projection}
\end{equation}
The matter and scalaron contributions therefore enter the radial Ricci
projection additively. Consequently, the sign of
$T_{\mu\nu}K^\mu K^\nu$ alone does not determine radial null
convergence in a generic $f(R)$ configuration~\cite{Capozziello:2013vna,Capozziello:2014bqa}.

\section{Radial null geometry and the exact identity}
\label{sec:null-geometry}

\subsection{Null frame and expansions}

On a static interval, define the future-directed outgoing and ingoing
radial null vectors
\begin{align}
 K^\mu&=\left(\frac{1}{\sqrt{A}},\sqrt{B},0,0\right),
 \label{eq:K-vector}\\
 L^\mu&=\frac12\left(\frac{1}{\sqrt{A}},-\sqrt{B},0,0\right).
 \label{eq:L-vector}
\end{align}
With the signature adopted in Eq.~\eqref{eq:metric}, they satisfy
\begin{equation}
 K^\mu K_\mu=L^\mu L_\mu=0,
 \qquad
 K^\mu L_\mu=1.
\end{equation}
This null frame is well defined on the open static interval, but it
generally does not extend regularly to a boundary at which \(A=0\).
Any limiting surface must therefore be interpreted using a regular
extension of the geometry rather than the normalization of
Eqs.~\eqref{eq:K-vector} and~\eqref{eq:L-vector} alone.

For any radial null vector \(N^\mu\), the expansion of the associated
congruence is obtained from the variation of the area
\(\mathcal{A}=4\pi r^2\) of the symmetry spheres,
\begin{equation}
 \theta_{(N)}
 =\frac{1}{\mathcal{A}}\mathcal{L}_{N}\mathcal{A}
 =\frac{2}{r}N^\mu\nabla_\mu r .
 \label{eq:general-expansion}
\end{equation}
Consequently, within the static interval,
\begin{equation}
 \theta_{(K)}=\frac{2\sqrt{B}}{r},
 \qquad
 \theta_{(L)}=-\frac{\sqrt{B}}{r}.
 \label{eq:expansions}
\end{equation}

As \(B\rightarrow0^+\), the areal-radius derivatives along this
particular static frame tend to zero. A zero of \(B\), however, does
not by itself determine the geometric or causal nature of the limiting
surface. If \(A\) remains finite and nonzero and the spacetime admits
a regular extension, the surface may correspond to a non-Killing
minimal sphere, such as a static wormhole throat~\cite{Morris:1988cz}.
 Modified-gravity wormhole solutions provide explicit settings
in which the effective curvature sector supports such geometries even
when the ordinary matter sector is nonexotic~\cite{Harko:2013yb}.
If \(A\) and \(B\) possess a regular common zero, the boundary may
instead be a Killing horizon of the static geometry. No classification as an event,
cosmological, or Cauchy horizon follows from the static frame alone.

The vector \(K^\mu\) is tangent to radial null geodesics but is not
affinely parametrized:
\begin{equation}
 K^\nu\nabla_\nu K^\mu=\kappa_K K^\mu,
 \qquad
 \kappa_K=\frac{\sqrt{B}\,A'}{2A}.
 \label{eq:nonaffinity}
\end{equation}
Within the open static interval, a positive regular rescaling of the
null direction multiplies both \(\theta_{(K)}\) and
\(R_{\mu\nu}K^\mu K^\nu\) by positive factors and therefore leaves
their signs unchanged. The normalization in
Eq.~\eqref{eq:K-vector} is adopted because it leads directly to the
simple radial identity derived below. Common-horizon limits will be
treated separately through the finite one-sided limit of \(B/A\).

\subsection{Raychaudhuri identity}

For a hypersurface-orthogonal null congruence, the non-affine
Raychaudhuri equation is
\begin{align}
 K^\mu\nabla_\mu\theta_{(K)}
 ={}&\kappa_K\theta_{(K)}
 -\frac12\theta_{(K)}^2
 -\sigma_{\mu\nu}\sigma^{\mu\nu}\notag\\
 &+\omega_{\mu\nu}\omega^{\mu\nu}
 -R_{\mu\nu}K^\mu K^\nu .
 \label{eq:Raychaudhuri}
\end{align}
The twist (vorticity) $\omega_{\mu\nu}$ vanishes by hypersurface orthogonality, while the shear
$\sigma_{\mu\nu}$ vanishes for a radial null congruence in spherical symmetry.
Using Eq.~\eqref{eq:expansions}, one obtains
\begin{equation}
 K^\mu\nabla_\mu\theta_{(K)}
 =\frac{B'}{r}-\frac{2B}{r^2}.
\end{equation}
Substitution into Eq.~\eqref{eq:Raychaudhuri} gives the purely
geometric identity
\begin{equation}
 R_{\mu\nu}K^\mu K^\nu
 =-\frac{B'}{r}+\frac{BA'}{rA}.
 \label{eq:RKK-geometric}
\end{equation}
This result is independent of the gravitational field equations.

For a radial scalar $\F(r)$, direct evaluation of its Hessian gives
\begin{equation}
 K^\mu K^\nu\nabla_\mu\nabla_\nu\F
 =B\left[
 \F''+\left(\frac{B'}{2B}-\frac{A'}{2A}\right)\F'
 \right].
 \label{eq:scalaron-Hessian}
\end{equation}
Combining Eqs.~\eqref{eq:Ricci-projection}
and~\eqref{eq:RKK-geometric} yields
\begin{equation}
 B'-\frac{A'}{A}B
 =-r\frac{\Eeff}{\F}.
 \label{eq:first-order-B}
\end{equation}

Introducing the radial metric ratio
\begin{equation}
 \Q(r)\equiv\frac{B(r)}{A(r)},
 \label{eq:Q-definition}
\end{equation}
we arrive at the central identity
\begin{equation}
 \Q'(r)
  =-\frac{r}{A(r)\F(r)}\,\Eeff(r).
 \label{eq:master-identity}
\end{equation}
Equation~\eqref{eq:master-identity} is exact on every connected static
interval satisfying Eq.~\eqref{eq:static-domain}. It follows from the
metric \(f(R)\) field equations together with the radial geometric
identity~\eqref{eq:RKK-geometric}.

Under a constant rescaling of the static time coordinate,
\(\Q\) is multiplied by a positive constant. Its monotonicity,
constancy, and the ordering or equality of its endpoint values are
therefore independent of this residual normalization freedom.
Equation~\eqref{eq:master-identity} forms the basis of the subsequent
monotonicity and boundary results.

The ratio $\Q$ compares the two
independent metric potentials at fixed areal radius: $A(r)$ fixes the
local redshift, through the norm of the static Killing vector, while
$B(r)$ fixes the radial proper-distance element and, through
Eq.~\eqref{eq:expansions}, the areal-radius variation of the radial
null congruences. Thus $\Q$ is not merely an auxiliary algebraic
combination. At finite, nonzero $A$, its vanishing is equivalent to
$B\to0$ and hence to the vanishing of the expansions in the static
null frame; the geometric interpretation of the limiting surface
still requires a regular extension, as discussed above. At a regular
Killing horizon, where $A$ and $B$ vanish together, the finite
nonzero limit of $\Q$ instead measures their relative rate of
vanishing. On the open static interval one may equivalently write
\(\Q'/\Q=B'/B-A'/A=-r\Eeff/(B\F)\), so the effective radial
convergence directly controls the mismatch between the logarithmic
radial variations of the two metric potentials. In Einstein gravity
this mismatch is governed by $\rho+p_r$, whereas in metric $f(R)$
gravity the scalaron Hessian provides an additional geometric
contribution.

\section{Monotonicity and boundary theorems}
\label{sec:theorems}

The regularity inherited from the metric functions and the scalaron
will be taken to make \(\Eeff\) continuous on every compact
subinterval of the static domain. This assumption is stronger than is
needed for the non-strict statements, but it makes the strict
alternatives unambiguous.

\begin{proposition}[Radial monotonicity]
	\label{prop:monotonicity}
	Let \(I\) be a connected static interval satisfying
	Eq.~\eqref{eq:static-domain}. If \(\Eeff\geq0\) on \(I\), then
	\(\Q=B/A\) is nonincreasing. If \(\Eeff\leq0\) on \(I\), then \(\Q\)
	is nondecreasing. More precisely, for any \(r_1<r_2\) in \(I\), if
	\(\Eeff\geq0\) on \([r_1,r_2]\) and \(\Eeff\) is not identically zero
	there, then
	\begin{equation}
		\Q(r_2)<\Q(r_1).
	\end{equation}
	Likewise, if \(\Eeff\leq0\) and is not identically zero on the same
	subinterval, then \(\Q(r_2)>\Q(r_1)\). Equivalently, without assuming
	continuity, strictness follows whenever \(\Eeff\) has the indicated
	strict sign on a set of nonzero measure.
\end{proposition}

\begin{proof}
	On \(I\), the factor \(r/(A\F)\) is strictly positive.
	Equation~\eqref{eq:master-identity} therefore fixes the sign of
	\(\Q'\). Integration over \([r_1,r_2]\) gives
	\begin{equation}
		\Q(r_2)-\Q(r_1)
		=-\int_{r_1}^{r_2}\frac{r}{A\F}\,\Eeff\,dr .
	\end{equation}
	Continuity and nontriviality imply that \(\Eeff\) has the relevant
	strict sign on a nonempty open subinterval, and hence on a set of
	positive measure, which makes the integral strictly nonzero.
\end{proof}

\begin{proposition}[Boundary identity]
	\label{prop:boundary}
	Let \(I=(a,b)\) be a connected static interval, or let \(a<b\) be two
	points in the closure of such an interval, and suppose that \(\Q\)
	admits finite one-sided limits at both endpoints. Then
	\begin{equation}
		\Q(b)-\Q(a)
		=-\int_a^b\frac{r}{A\F}\,\Eeff\,dr .
		\label{eq:boundary-identity}
	\end{equation}
	Here \(\Q(a)\) and \(\Q(b)\) denote the corresponding one-sided
	limits. If an endpoint is not contained in the open static interval,
	the integral is understood as a signed improper one-sided integral;
	absolute convergence is not required by the proposition.
\end{proposition}

\begin{proof}
	For \(x,y\in I\), with \(a<x<y<b\), integration of
	Eq.~\eqref{eq:master-identity} gives
	\begin{equation}
		\Q(y)-\Q(x)
		=-\int_x^y\frac{r}{A\F}\,\Eeff\,dr .
	\end{equation}
	Taking \(x\to a^+\) and \(y\to b^-\), and using the assumed finite
	limits of \(\Q\), yields Eq.~\eqref{eq:boundary-identity}. The same
	limiting procedure defines the signed improper integral.
\end{proof}

\subsection{Finite-redshift zeros of the radial metric function}

Suppose that the endpoints \(a<b\) of a connected static interval
satisfy
\begin{equation}
	\lim_{r\to a^+}B(r)=\lim_{r\to b^-}B(r)=0,
\end{equation}
while \(A\) admits finite, strictly positive one-sided limits there.
Then
\begin{equation}
	\Q(a)=\Q(b)=0,
\end{equation}
and Eq.~\eqref{eq:boundary-identity} reduces to
\begin{equation}
	\int_a^b\frac{r}{A\F}\,\Eeff\,dr=0.
	\label{eq:zero-integral-nonkilling}
\end{equation}

If the limiting surfaces admit regular extensions as marginal
symmetry spheres, they may be interpreted as non-Killing marginal
boundaries. This interpretation is an additional geometric
assumption and does not follow from \(B\to0\) alone.

\begin{corollary}[Two finite-redshift zeros of \(B\)]
	\label{cor:nonkilling}
	Let a connected static interval be bounded by two zeros of \(B\) at
	which \(A\) has finite, strictly positive one-sided limits. If
	\(B>0\) in the interior and \(\Eeff\) is continuous, then \(\Eeff\)
	cannot be sign-definite. It must assume both positive and negative
	values in the interval.
\end{corollary}

\begin{proof}
	The function \(\Q=B/A\) is positive in the interior and tends to zero
	at both endpoints, so it cannot be monotonic on the complete interval.
	By Proposition~\ref{prop:monotonicity}, a nonnegative or nonpositive
	\(\Eeff\) would make it monotonic. Hence neither fixed sign is
	possible. Continuity then implies that \(\Eeff\) takes both signs.
\end{proof}

\subsection{Static patches bounded by Killing horizons}

For clarity, a \emph{regular nondegenerate static-patch Killing
	horizon} will mean an endpoint \(r_i\) admitting a regular extension
for which
\begin{align}
	A(r_i)&=B(r_i)=0,
	\qquad A'(r_i)B'(r_i)>0,\notag\\
	0&<\lim_{r\to r_i}\F(r)<\infty
	\label{eq:regular-horizon-definition}
\end{align}
and both metric functions have simple one-sided zeros when approached
from the static patch. Thus,
\begin{align}
	A(r)&=A'(r_i)(r-r_i)
	+O\bigl((r-r_i)^2\bigr),\notag\\
	B(r)&=B'(r_i)(r-r_i)
	+O\bigl((r-r_i)^2\bigr).
	\label{eq:horizon-expansions}
\end{align}
The corresponding endpoint datum is
\begin{equation}
	q_i\equiv
	\lim_{r\to r_i}\frac{B(r)}{A(r)}
	=\frac{B'(r_i)}{A'(r_i)}>0 .
	\label{eq:horizon-q}
\end{equation}
The value \(q_i\) depends on the constant normalization of the static
time coordinate, but all \(q_i\) acquire the same positive factor.
Their equality, ordering, and ratios are therefore independent of
that residual freedom.

For two such horizons \(r_1<r_2\) bounding a connected static patch,
Proposition~\ref{prop:boundary} gives
\begin{equation}
	q_2-q_1
	=-\int_{r_1}^{r_2}\frac{r}{A\F}\,\Eeff\,dr .
	\label{eq:horizon-slope-identity}
\end{equation}
The integral is understood in the signed improper sense. Its existence
follows from the finite one-sided limits of \(\Q\).

\begin{corollary}[Horizon-slope ordering]
	\label{cor:horizon-ordering}
	Let two regular, nondegenerate Killing horizons bound a connected
	static patch. If \(\Eeff\geq0\) throughout the patch, then
	\(q_2\leq q_1\). If \(\Eeff\leq0\), then \(q_2\geq q_1\). The
	inequality is strict if \(\Eeff\) is not identically zero.
\end{corollary}

\begin{proof}
	The result follows by applying Proposition~\ref{prop:monotonicity} and
	taking the one-sided limits of \(\Q\) at the two horizons.
\end{proof}

\begin{corollary}[Equal-endpoint rigidity]
	\label{cor:equal-endpoints}
	Suppose that \(q_1=q_2\) and that \(\Eeff\) is everywhere
	nonnegative or everywhere nonpositive in the static patch. Then
	\begin{equation}
		\Eeff\equiv0,
		\qquad
		\Q\equiv q_1 .
		\label{eq:equal-endpoint-rigidity}
	\end{equation}
	In particular, a fixed-sign, nonvanishing \(\Eeff\) is incompatible
	with equal limiting values of \(B/A\) at the two horizons.
\end{corollary}

\begin{proof}
	Under either sign assumption, \(\Q\) is monotonic. A monotonic
	function with equal finite endpoint limits is constant. The master
	identity then implies \(\Eeff=0\) throughout the static patch.
\end{proof}

\begin{proposition}[Saturation class]
	\label{prop:saturation}
	On a connected static interval,
	\begin{equation}
		\Eeff\equiv0
		\quad\Longleftrightarrow\quad
		\frac{B}{A}=C,
		\label{eq:saturation-equivalence}
	\end{equation}
	where \(C>0\) is constant.
\end{proposition}

\begin{proof}
	If \(\Eeff\equiv0\), Eq.~\eqref{eq:master-identity} gives
	\(\Q'=0\), and connectedness implies \(\Q=C\). Conversely, if
	\(B/A=C\), then \(\Q'=0\), and the same identity, together with
	\(r>0\), \(A>0\), and \(\F>0\), gives \(\Eeff\equiv0\). Positivity of
	\(A\) and \(B\) implies \(C>0\).
\end{proof}

\begin{remark}
	These are monotonicity and boundary theorems, not an unconditional
	count of Killing horizons. A sign condition on \(\Eeff\) constrains
	the ordering of the finite endpoint data \(q_i\); it excludes two
	regular horizons only after an additional condition, such as
	\(q_1=q_2\), removes the boundary difference.
\end{remark}

{
	The complementary roles of bulk monotonicity and boundary data are
	summarized schematically in Fig.~\ref{fig:monotonicity-summary}.
	Panel~(a) illustrates the direct relation between the sign of
	\(\Eeff\) and the radial behavior of \(\Q=B/A\).
	Panel~(b) shows that, for a static patch bounded by regular
	nondegenerate Killing horizons, the same sign condition orders the
	finite endpoint values \(q_i\). Panel~(c) illustrates the distinct
	finite-redshift case in which \(\Q\) vanishes at both endpoints while
	remaining positive in the interior; the required change from
	increasing to decreasing behavior then forces a continuous
	\(\Eeff\) to take both signs.
}

\begin{figure*}[th!]
	\centering
	\includegraphics[width=0.9\textwidth]
	{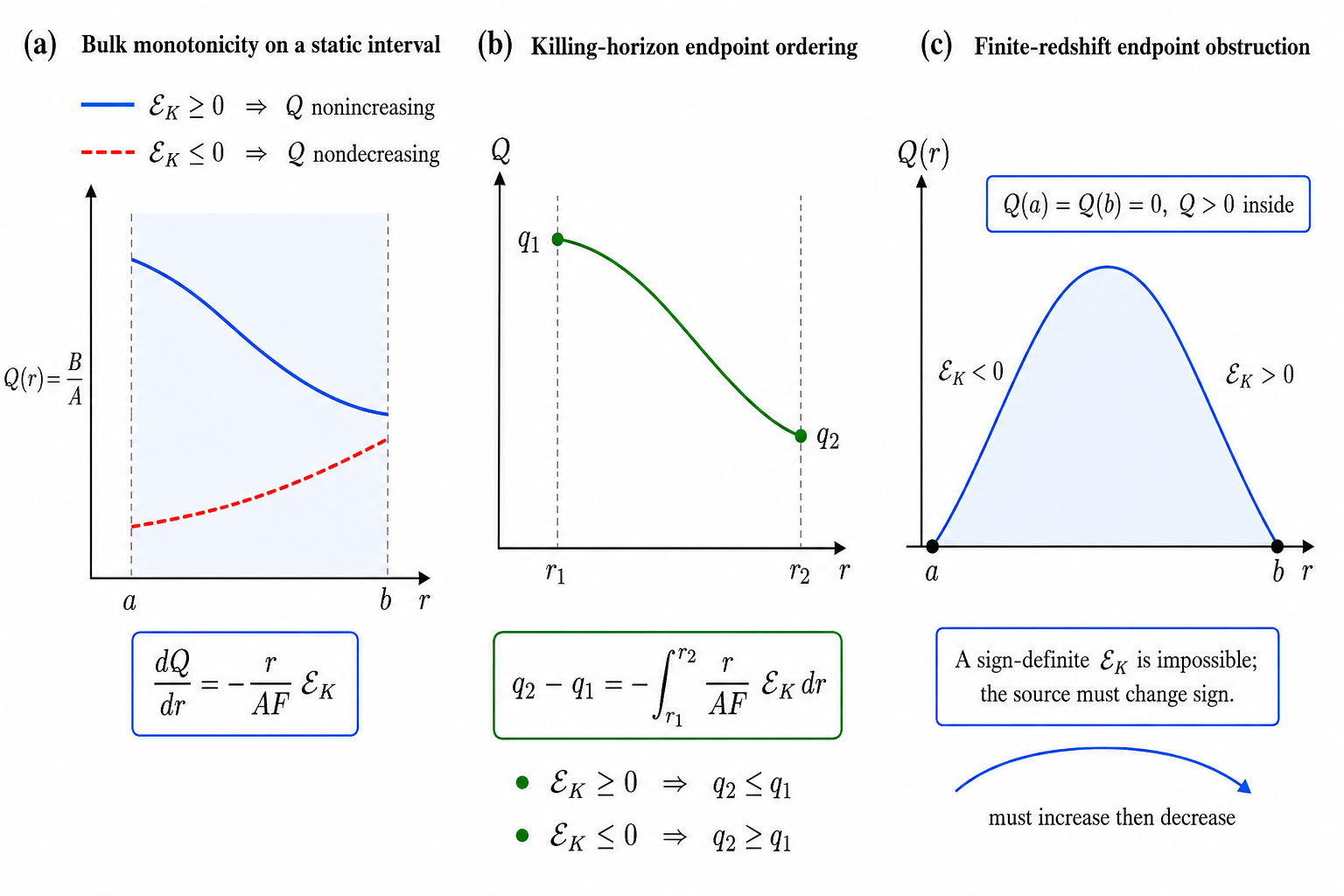}
	\caption{{
			Schematic summary of the radial monotonicity and boundary results.
			Panel~(a) illustrates the master identity
			\(\Q'=-(r/A\F)\Eeff\): on a connected static interval,
			\(\Eeff\geq0\) makes \(\Q=B/A\) nonincreasing, whereas
			\(\Eeff\leq0\) makes it nondecreasing.
			Panel~(b) shows the corresponding ordering of the finite endpoint
			data \(q_i=\lim B/A\) for a static patch bounded by regular
			nondegenerate Killing horizons. In particular,
			\(\Eeff\geq0\) implies \(q_2\leq q_1\), while
			\(\Eeff\leq0\) implies \(q_2\geq q_1\), in accordance with
			Eq.~\eqref{eq:horizon-slope-identity}.
			Panel~(c) illustrates the distinct case of two finite-redshift zeros
			of \(B\), for which \(\Q(a)=\Q(b)=0\) while \(\Q>0\) in the
			interior. The function must therefore increase somewhere and decrease
			somewhere, so a continuous \(\Eeff\) cannot be sign-definite and must
			assume both positive and negative values, as stated in
			Corollary~\ref{cor:nonkilling}. The curves are schematic and are
			intended to display the sign, monotonicity, and endpoint structure
			rather than a particular solution.}}
	\label{fig:monotonicity-summary}
\end{figure*}

\section{Matter and scalaron contributions}
\label{sec:physical-content}

Introduce the orthonormal static vectors
\begin{equation}
u^\mu=\left(\frac{1}{\sqrt A},0,0,0\right),
\qquad
s^\mu=\left(0,\sqrt B,0,0\right),
\end{equation}
which satisfy \(u^\mu u_\mu=1\), \(s^\mu s_\mu=-1\), and
\(u^\mu s_\mu=0\). The outgoing null vector is
\(K^\mu=u^\mu+s^\mu\). For a static anisotropic source with no radial
energy flux, energy density \(\rho\), and radial pressure \(p_r\),
\begin{equation}
T_{\mu\nu}K^\mu K^\nu=\rho+p_r.
\label{eq:matter-projection}
\end{equation}
Using Eq.~\eqref{eq:scalaron-Hessian}, the effective convergence
numerator becomes
\begin{equation}
\Eeff=8\pi(\rho+p_r)+B\F''
+\frac12\left(B'-\frac{BA'}{A}\right)\F'.
\label{eq:Eeff-anisotropic}
\end{equation}
The last term measures the mismatch between the logarithmic
derivatives of the radial and redshift functions, since its
coefficient is \(B(B'/B-A'/A)/2\). It vanishes in the
proportional-metric class \(B=CA\).

If \(f(R)\) is at least three times differentiable over the relevant
curvature range, then
\begin{equation}
\F'=f_{RR}R',
\qquad
\F''=f_{RR}R''+f_{RRR}(R')^2,
\label{eq:chain-rule}
\end{equation}
and therefore
\begin{align}
\Eeff={}&8\pi(\rho+p_r)
+B\Bigg[f_{RR}R''+f_{RRR}(R')^2\notag\\
&+\left(\frac{B'}{2B}-\frac{A'}{2A}\right)f_{RR}R'\Bigg].
\label{eq:Eeff-curvature}
\end{align}
The sign of \(\Eeff\) is controlled neither by \(R''\) nor by
\(f_{RR}\) alone, but by their combination with the scalaron profile,
the matter projection, and the local metric geometry.

\subsection{General-Relativity limit}

In the General-Relativity limit, including a possible cosmological
constant, \(\F=1\) and \(\F'=0\), so
\begin{equation}
\Eeff=8\pi(\rho+p_r).
\label{eq:GR-limit}
\end{equation}
The radial null energy condition determines the monotonicity of
\(B/A\). Vacuum and a purely radial Maxwell field satisfy
\(\rho+p_r=0\) and hence saturate the radial projection. By
Proposition~\ref{prop:saturation}, \(B/A=C\) on each connected static
interval, and a constant rescaling of the static time coordinate may
set \(C=1\). This reproduces the familiar gauge \(A=B\), but the
radial null equation alone does not determine the complete geometry
or its horizon structure.

This limit agrees with the recent analysis of static spherical energy
conditions by Wang and Battista~\cite{WangBattista2026}. Their radial
NEC condition makes
\(A_{\rm WB}B_{\rm WB}=A/B=1/\Q\) nondecreasing, which is equivalent
to the nonincreasing behavior of \(\Q\) obtained here when
\(\F=1\). Equations~\eqref{eq:Eeff-anisotropic} and
\eqref{eq:Eeff-curvature} display the additional scalaron source in
metric \(f(R)\) gravity.

Table~\ref{tab:wb_comparison} summarizes the principal points of
contact and departure from Ref.~\cite{WangBattista2026}.

\begin{table}[t]
\centering
\caption{Comparison with the static radial energy-condition analysis of
Wang and Battista~\cite{WangBattista2026}.}
\label{tab:wb_comparison}
\begin{tabular}{lll}
\hline\hline
 & \textbf{Wang--Battista} & \textbf{Present work} \\
\hline
Theory & General Relativity & Metric $f(R)$ gravity \\
Source & $\rho+p_r$ & $\Eeff$ (matter $+$ scalaron) \\
Quantity & $A_{\rm WB}B_{\rm WB}=1/\Q$ & $\Q=B/A$ \\
Boundary & NEC ordering & Slope ratios $q_i$ \\
Scope & Static GR & Static $f(R)$ \\
\hline\hline
\end{tabular}
\end{table}

As shown above, the two results coincide exactly in the $\F=1$ limit,
with $\Q$ and $1/\Q$ tracking the same monotonicity up to the
inversion implied by the different choice of ratio. The present
identity generalizes this result to metric $f(R)$ gravity and, in
addition, isolates explicitly the finite Killing-horizon endpoint
data $q_i$ relevant to the integrated radial identity.

\subsection{A non-saturated exact benchmark: the constant-density star}
\label{sec:constant-density}

A complete exact example with \(\F>0\), \(\Eeff\neq0\), and
nonconstant \(\Q\) is provided by the Schwarzschild interior solution
for a constant-density perfect fluid~\cite{Schwarzschild1916}. Let
\(M\) be the total mass and \(r_\star\) the stellar radius. For
\(0\leq r\leq r_\star\), define
\begin{align}
B(r)&=1-\frac{2Mr^2}{r_\star^3},\notag\\
A(r)&=\frac14\left[3\sqrt{1-\frac{2M}{r_\star}}
-\sqrt{1-\frac{2Mr^2}{r_\star^3}}\right]^2 .
\label{eq:interior-Schwarzschild-metric}
\end{align}
Regularity and finite central pressure require
\(r_\star>9M/4\). Writing
\(B_\star\equiv B(r_\star)=1-2M/r_\star\), the density and isotropic
pressure are
\begin{align}
\rho&=\frac{3M}{4\pi r_\star^3},\notag\\
p(r)&=\rho\,
\frac{\sqrt{B(r)}-\sqrt{B_\star}}
{3\sqrt{B_\star}-\sqrt{B(r)}} .
\label{eq:interior-Schwarzschild-matter}
\end{align}
For \(f(R)=R\), one has \(\F=1\) and
\begin{equation}
\Eeff=8\pi\bigl[\rho+p(r)\bigr]>0
\qquad (0<r<r_\star).
\label{eq:interior-Schwarzschild-Eeff}
\end{equation}
The metric ratio is
\begin{equation}
\Q(r)=
\frac{4B(r)}{\left[3\sqrt{B_\star}-\sqrt{B(r)}\right]^2},
\label{eq:interior-Schwarzschild-Q}
\end{equation}
and direct differentiation gives
\begin{equation}
\Q'(r)=
-\frac{48Mr\sqrt{B_\star}}
{r_\star^3\left[3\sqrt{B_\star}-\sqrt{B(r)}\right]^3}<0
\label{eq:interior-Schwarzschild-Qprime}
\end{equation}
for \(0<r<r_\star\). Thus the solution realizes the strict branch of
Proposition~\ref{prop:monotonicity}. It decreases from
\begin{equation}
\Q(0)=\frac{4}{\left(3\sqrt{B_\star}-1\right)^2}
\end{equation}
to \(\Q(r_\star)=1\), where it matches the exterior Schwarzschild
geometry in the standard normalization.

This benchmark is deliberately conservative: it is a complete exact
solution rather than a reconstruction based on a subset of equations.
Although the incompressible equation of state is physically idealized,
the solution provides a fully analytic and internally complete
benchmark for the strict monotonicity theorem. It demonstrates that
the theorem has nontrivial content away from the saturation class and
provides a reference against which genuinely modified-gravity examples
can be tested.

\section{Equality classes and exact-solution diagnostics}
\label{sec:applications}

\subsection{The proportional-metric vacuum class}
\label{sec:proportional-class}

The saturation class has a sharp vacuum consequence.

\begin{proposition}[Vacuum scalaron rigidity]
\label{prop:vacuum-rigidity}
Let a vacuum solution of metric \(f(R)\) gravity satisfy
\begin{equation}
B(r)=C A(r),
\qquad C>0,
\label{eq:proportional-metric}
\end{equation}
on a connected static interval. Then
\begin{equation}
\F''(r)=0,
\qquad
\F(r)=F_0+F_1r,
\label{eq:linear-scalaron}
\end{equation}
where \(F_0\) and \(F_1\) are constants on that interval. This is a
necessary consequence of the radial null projection, but is not by
itself sufficient to solve the full modified field equations.
\end{proposition}

\begin{proof}
For \(B=CA\), Eq.~\eqref{eq:RKK-geometric} gives
\(R_{\mu\nu}K^\mu K^\nu=0\), while
\(B'/B=A'/A\) makes the connection term in
Eq.~\eqref{eq:scalaron-Hessian} vanish. The vacuum null projection
therefore reduces to \(B\F''=0\). Since \(B>0\) in the static
interval, \(\F''=0\), and integration gives
Eq.~\eqref{eq:linear-scalaron}.
\end{proof}

This statement geometrically repackages the constant-\(X\) radial
equation already present in the construction of Multam\"aki and
Vilja~\cite{MultamakiVilja2006}. Its usefulness here is that it arises
as the equality case of the general matter-plus-scalaron monotonicity
law and therefore remains meaningful independently of the particular
reconstruction scheme.

In the presence of a static anisotropic source with no radial energy
flux, the proportional-metric condition instead gives
\begin{equation}
8\pi(\rho+p_r)+B\F''=0.
\label{eq:proportional-matter-balance}
\end{equation}
The matter and scalaron projections need not vanish separately; they
must cancel exactly.

Constant-curvature vacuum solutions have \(\F'=0\) and belong to the
special case \(F_1=0\). Schwarzschild--(anti-)de Sitter-type solutions
written in the gauge \(A=B\) therefore saturate the radial identity.
This fact does not determine their number of horizons. Their horizon
structure and thermodynamics require the remaining field equations
and global data, as illustrated by the broad literature on
constant-curvature, topological, and electrodynamic \(f(R)\) black
holes~\cite{Cognola:2015wqa,EslamPanah2023,Xu2024}.

\subsection{Constant-\texorpdfstring{\(X\)}{X} solutions of Multam\"aki and Vilja}
\label{sec:MV-constant-X}

Multam\"aki and Vilja~\cite{MultamakiVilja2006} use
\begin{equation}
ds^2=s(r)\,dt^2-p(r)\,dr^2-r^2d\Omega^2,
\end{equation}
with
\begin{equation}
A=s,
\qquad
B=\frac{1}{p},
\qquad
X(r)\equiv p(r)s(r)=\frac{A}{B}.
\label{eq:MV-dictionary}
\end{equation}
Hence
\begin{equation}
\Q=\frac{B}{A}=\frac{1}{X}.
\label{eq:Q-X}
\end{equation}
Whenever \(X=X_0\) is constant, \(\Q=1/X_0\) is constant. On every
connected static subinterval satisfying \(\F>0\), the solution belongs
to the saturation class, and the vacuum radial equation gives
\begin{equation}
\F''=0,
\qquad
\F(r)=F_0+F_1r.
\end{equation}

Two explicit nonconstant-curvature solutions displayed in
Ref.~\cite{MultamakiVilja2006} have \(X=2\). In the present notation,
\begin{align}
\text{(I)}\quad
&A(r)=1-\frac{2M}{r}-\frac{\Lambda r^2}{3},
&\F(r)&=1-\frac{r}{3M},
\label{eq:MV-I}\\
\text{(II)}\quad
&A(r)=1-\frac{\Lambda r^2}{3},
&\F(r)&=F_0r,
\label{eq:MV-II}
\end{align}
with
\begin{equation}
B(r)=\frac{A(r)}{2},
\qquad
R(r)=-\frac{1}{r^2}-2\Lambda .
\label{eq:MV-curvature}
\end{equation}
Although the redshift function of solution (I) has the
Schwarzschild--de Sitter form, the complete metric is not standard
Schwarzschild--de Sitter in the displayed normalization because
\(B=A/2\).

For both solutions,
\begin{equation}
\Q=\frac12,
\qquad
\Eeff=0.
\end{equation}
Indeed, \(\F''=0\) and \(B'/B=A'/A\), so the scalaron Hessian in
Eq.~\eqref{eq:scalaron-Hessian} vanishes. They realize the equality
class on every connected static subinterval satisfying the assumptions
of the theorem.

Solution (I) illustrates why \(\F>0\) must be checked throughout any
interval used in a boundary argument. For
\begin{equation}
0<9M^2\Lambda<1,
\label{eq:SdS-range}
\end{equation}
the redshift function has a black-hole root and a cosmological root
and is positive between them. At \(r=3M\),
\begin{equation}
A(3M)=\frac{1-9M^2\Lambda}{3}>0,
\qquad
\F(3M)=0.
\label{eq:F-zero-MV}
\end{equation}
Thus the zero of \(\F\) lies inside the two-horizon static patch. At
this surface the effective-coupling interpretation proportional to
\(1/\F\) and the usual Einstein-frame conformal transformation becomes singular, and
the field equations cannot be divided by \(\F\). The undivided radial
projection remains algebraically satisfied here because both sides
vanish, so no curvature singularity follows from this observation
alone. Nevertheless, the positive-coupling hypothesis and the divided
master identity fail across the surface. The complete patch therefore
cannot serve as a global example of the monotonicity theorem, although
each subinterval with \(\F>0\) remains admissible.

For solution (II), \(F_0>0\) and \(r>0\) imply \(\F>0\) wherever the
metric is static. It provides a direct realization of the equality
class, but
\begin{equation}
R(r)\sim-\frac{1}{r^2}
\qquad (r\to0^+)
\end{equation}
shows that saturation does not imply regularity.

\subsection{Power-law radial consistency test}
\label{sec:power-law}

Consider the power-law ansatz
\begin{equation}
A(r)=s_0r^m,
\quad
\F(r)=F_0r^n,
\quad
X(r)=\chi A(r),
\label{eq:power-law-ansatz}
\end{equation}
where \(X=A/B\) and \(s_0\), \(F_0\), and \(\chi\) are nonzero
constants. Then
\begin{equation}
B=\frac{1}{\chi},
\qquad
\Q=\frac{1}{\chi s_0r^m}.
\label{eq:power-law-Q}
\end{equation}
On a static interval with \(r>0\), positivity requires
\(s_0>0\), \(\chi>0\), and \(F_0>0\).

Since \(B\) is constant,
\begin{equation}
R_{\mu\nu}K^\mu K^\nu
=\frac{Bm}{r^2},
\label{eq:power-law-RKK}
\end{equation}
whereas
\begin{equation}
K^\mu K^\nu\nabla_\mu\nabla_\nu\F
=BF_0n\left(n-1-\frac{m}{2}\right)r^{n-2}.
\label{eq:power-law-Hessian}
\end{equation}
The vacuum radial equation therefore requires
\begin{equation}
m(n+2)=2n(n-1).
\label{eq:power-law-relation}
\end{equation}
This is a necessary radial condition; the remaining field equations
must still be imposed.

When Eq.~\eqref{eq:power-law-relation} holds,
\begin{equation}
\Eeff=BF_0m\,r^{n-2},
\qquad
\Q'=-\frac{m}{r}\Q,
\label{eq:power-law-check}
\end{equation}
in exact agreement with Eq.~\eqref{eq:master-identity}. Thus \(\Q\)
decreases for \(m>0\), increases for \(m<0\), and is constant for
\(m=0\).

The exponent relation displayed for solution III in
Ref.~\cite{MultamakiVilja2006},
\begin{equation}
n=\frac{2m(m-1)}{m-2},
\label{eq:MV-displayed-power-law}
\end{equation}
is not generically equivalent to
Eq.~\eqref{eq:power-law-relation}. Appendix~\ref{app:MV-power-law}
shows that the mismatch already appears upon direct substitution into
the original radial equation of that reference, independently of the
normalization coefficient \(\chi\). In view of the general consistency
issues that can arise when additional metric relations are imposed in
spherical \(f(R)\) reconstructions~\cite{SaffariRahvar2009}, the
published power-law family should be regarded as unresolved until the
exponent discrepancy is traced to a typographical error, an omitted
restriction, or a correction of the displayed solution. We do not
claim on the basis of the radial equation alone that the entire
reconstruction is invalid.

Since \(B=1/\chi\) is strictly positive and \(A=s_0r^m\) has no zero
at finite \(r>0\), this ansatz has neither a finite-radius zero of
\(B\) nor a finite-radius Killing horizon within its static domain.
Its role is therefore that of an algebraic consistency diagnostic,
not an independent horizon-exclusion example.

\section{Scope and relation to black-hole interiors}
\label{sec:scope}

The present framework applies to connected static intervals. This
restriction is geometric rather than merely technical. In a
Reissner--Nordstr\"om-type black hole, the region between the outer
event horizon and the inner Killing horizon is nonstatic: the areal
coordinate becomes timelike, and the conditions \(A>0\) and \(B>0\)
do not hold. The null frame of Eqs.~\eqref{eq:K-vector} and
\eqref{eq:L-vector} therefore does not provide a horizon-regular
description of that interior region.

Accordingly, Eqs.~\eqref{eq:master-identity} and
\eqref{eq:horizon-slope-identity} are not a no-Cauchy-horizon theorem.
They constrain the radial geometry within a static patch and provide
exact boundary information when that patch is bounded by regular
Killing horizons, as for a black-hole--cosmological-horizon domain.
They do not determine the global causal character of a second boundary
or establish that a marginal or inner Killing horizon is a Cauchy
horizon.

A genuine analysis of an inner Cauchy horizon requires a formulation
regular across the outer horizon. The present static result should be understood as the static reduction
of such a horizon-regular formulation rather than as a self-contained
alternative to it. Any consistent double-null treatment, when
specialized to a static spherically symmetric region and expressed in
the corresponding static normalization, should recover the content of
Eq.~\eqref{eq:master-identity} and
Propositions~\ref{prop:monotonicity}--\ref{prop:boundary}. In this
sense, the endpoint quantities \(q_i\) identified here provide natural
static-sector boundary data against which a horizon-regular extension
can be matched before addressing a possible second, inner horizon. With the
\((+---)\) signature, a convenient double-null form is
\begin{equation}
ds^2=2e^{-2\sigma(u,v)}\,du\,dv-r^2(u,v)d\Omega^2 .
\label{eq:double-null-metric}
\end{equation}
The two null expansions are proportional to \(\partial_v r\) and
\(\partial_u r\), respectively. The scalaron-modified Raychaudhuri and
cross-focusing equations can then be followed through the nonstatic
interior without relying on the singular normalization of a static
null frame at a horizon.

Any statement about a regular Cauchy horizon would require additional
assumptions linking the candidate global causal boundary to a regular
inner marginal or Killing horizon. Singular null boundaries, including
those associated with mass inflation, lie outside the present static
analysis.

The value of the static result is precisely this separation of scope.
It identifies the exact monotonic quantity, the necessary endpoint
terms, and the equality classes before any extension into the
black-hole interior is attempted. A future double-null treatment should
recover the master identity as its static reduction, but the present
paper supplies only the static-sector starting point for that broader
program.

\section{Conclusions}
\label{sec:conclusions}

We have established an exact radial monotonicity law for static,
spherically symmetric configurations in metric \(f(R)\) gravity. On a
connected interval with \(A>0\), \(B>0\), and \(f_R>0\), the sign of
the matter-plus-scalaron radial-convergence numerator fixes the
monotonicity of \(B/A\). In the General-Relativity limit this reduces
to the radial null-energy-condition result for the inverse metric
product, while in modified gravity the scalaron Hessian can reinforce,
cancel, or reverse the matter contribution.

The principal geometric point is the treatment of the interval
boundaries. A zero of \(B\) with finite nonzero \(A\) gives a vanishing
endpoint value of \(B/A\), so two such boundaries force the effective
convergence to change sign. At a regular nondegenerate Killing horizon,
however, the finite datum is the derivative ratio
\(B'(r_i)/A'(r_i)\). A fixed-sign convergence condition therefore
orders these endpoint values; it does not exclude two Killing horizons
unless an additional condition, such as equality of the two ratios,
removes the boundary difference.

Saturation is equivalent to a constant ratio \(B/A\). In vacuum this
requires a scalaron linear in the areal radius, which gives a geometric
interpretation of the constant-\(X\) equation used in earlier exact
reconstructions. The constant-density stellar interior supplies a
complete non-saturated benchmark: its positive radial convergence
makes \(B/A\) decrease strictly from the center to the stellar surface.
The Multam\"aki--Vilja applications then show both the utility and the
limitations of the diagnostic. One constant-\(X\) solution crosses
\(f_R=0\) inside its complete two-horizon static patch and hence falls
outside the positive-coupling theorem on that domain. For the displayed
power-law family, direct substitution into the original radial equation
produces an exponent condition different from the published one; the
status of that family remains unresolved pending a correction or an
additional restriction.

These results provide a practical hierarchy of checks for static
\(f(R)\) geometries: verify the sign and regularity of \(f_R\), compute
the effective radial convergence, retain the correct endpoint data,
and test any reconstructed ansatz against the undivided radial equation
before using the full field equations. Their scope remains the static
sector. Extending the analysis to inner horizons requires a
horizon-regular double-null system that treats both null expansions and
the scalaron dynamics through the nonstatic black-hole interior.

Beyond the horizon-regular double-null extension already outlined in
Sec.~\ref{sec:scope}, the present static framework admits several
further directions. A natural one is rotation: the loss of a single
areal radial coordinate in axisymmetric spacetimes requires
reformulating \(\Q\) in terms of the relevant null expansions of a
Kerr-like \(f(R)\) background, rather than as a simple ratio of two
metric functions. Another is the extension to general scalar-tensor
and Horndeski gravity, where the scalar field's kinetic and
non-minimal coupling terms enter the null projection in a manner
structurally similar to, but more involved than, the scalaron Hessian
of Eq.~\eqref{eq:scalaron-Hessian}. Identifying the corresponding
monotonic quantity and its Killing-horizon boundary data in these
broader theories would test whether the separation between local
radial monotonicity and global causal structure established here
persists beyond metric \(f(R)\) gravity.

\begin{acknowledgments}
FSNL acknowledges funding from the Funda\c{c}\~{a}o para a Ci\^{e}ncia e a Tecnologia (FCT) through national funds under the research grant UID/04434/2025 (DOI 10.54499/UID/04434/2025), and support from the FCT Scientific Employment Stimulus contract with reference CEECINST/00032/2018.
\end{acknowledgments}

\appendix

\section{Direct consistency check of the Multam\"aki--Vilja power-law family}
\label{app:MV-power-law}

The vacuum spherical equations of Multam\"aki and Vilja include the
radial relation~\cite{MultamakiVilja2006}
\begin{equation}
2\frac{X'}{X}
+r\frac{\F'}{\F}\frac{X'}{X}
-2r\frac{\F''}{\F}=0,
\label{eq:MV-original-radial}
\end{equation}
where \(X=ps=A/B\) and \(\F=f_R\). Whenever
\(2\F+r\F'\neq0\), this equation is equivalently
\begin{equation}
\frac{X'}{X}
=\frac{2r\F''}{2\F+r\F'}.
\label{eq:MV-radial-solved}
\end{equation}

For the power-law ansatz
\begin{equation}
A=s_0r^m,
\qquad
X=\chi A=\chi s_0r^m,
\qquad
\F=F_0r^n,
\end{equation}
one has
\begin{equation}
\frac{X'}{X}=\frac{m}{r},
\qquad
\frac{2r\F''}{2\F+r\F'}
=\frac{2n(n-1)}{(n+2)r},
\end{equation}
provided \(n\neq-2\). Equation~\eqref{eq:MV-radial-solved} therefore
requires
\begin{equation}
m(n+2)=2n(n-1),
\end{equation}
which is precisely Eq.~\eqref{eq:power-law-relation}. The exceptional
case \(n=-2\) does not evade the conclusion: direct substitution into
Eq.~\eqref{eq:MV-original-radial} leaves a nonzero term proportional to
\(\F''\), so no finite \(m\) solves the equation.

By contrast, solution III in the table of
Ref.~\cite{MultamakiVilja2006} displays
\begin{equation}
n=\frac{2m(m-1)}{m-2}.
\end{equation}
The two relations agree only on isolated parameter values, not as
identities. The mismatch is independent of the coefficient \(\chi\),
which cancels from \(X'/X\), and it arises before the remaining field
equations or the reconstructed \(f(R)\) are considered.

This check establishes a necessary inconsistency between the displayed
exponents and the original radial equation as printed. It does not by
itself determine whether the source is a typographical error in the
table, an omitted restriction, or a more substantial defect in the
reconstruction. Consistent with the caution advocated for constrained
spherical reconstructions in Ref.~\cite{SaffariRahvar2009}, the
appropriate conclusion is that the published power-law family requires
clarification and a full component-by-component verification.

\end{document}